\documentclass{amsart}

\usepackage[algo2e,lined,boxruled]{algorithm2e}
\usepackage{amssymb}
\usepackage{xcolor}
\usepackage{subcaption}
\usepackage{amsmath}
\usepackage{array}
\usepackage{tikz}\usetikzlibrary{shapes.geometric,backgrounds,patterns,decorations.pathmorphing, decorations.pathreplacing, decorations.shapes}
\usepackage{caption}
\usepackage[colorlinks=true, linkcolor=blue, citecolor=blue]{hyperref}
\usepackage[english]{babel}
\usepackage[autostyle]{csquotes}
\usepackage{float}
\usepackage{mathtools}
\usepackage{enumitem}
\usepackage{todonotes}

\usepackage[T1]{fontenc}

\newtheorem{definition}{Definition}[section]
\newtheorem{theorem}{Theorem}[section]
\newtheorem{lemma}{Lemma}[section]

\newtheorem{proposition}{Proposition}[section]

\newtheorem{claim}{Claim}[section]

\usepackage{algorithm}
\usepackage{algorithmic}

\begin{document}

\title{On the complexity of co-secure dominating set problem}
\author{B S Panda \and Soumyashree Rana}
\address{Department of Mathematics, Indian Institute of Technology Delhi, India} 
\email{bspanda@maths.iitd.ac.in, maz218122@iitd.ac.in}
\author{Sounaka Mishra}
\address{Department of Mathematics, Indian Institute of Technology Madras, India} 
\email{sounak@iitm.ac.in}

%


\maketitle

\begin{abstract}
A set $D \subseteq V$ of a graph $G=(V, E)$ is a dominating set of $G$ if every vertex $v\in V\setminus D$ is adjacent to at least one vertex in $D.$ A set $S \subseteq V$ is a co-secure dominating set (\textsc{CSDS}) of a graph $G$ if $S$ is a dominating set of $G$ and for each vertex $u \in S$ there exists a vertex $v \in V\setminus S$ such that $uv \in E$ and $(S\setminus \{u\}) \cup \{v\}$ is a dominating set of $G$ . The minimum cardinality of a co-secure dominating set of $G$ is the co-secure domination number and it is denoted by $\gamma_{cs}(G)$. Given a graph $G=(V, E)$, the minimum co-secure dominating set problem (\textsc{Min Co-secure Dom})  is to find a co-secure dominating set of minimum cardinality. 
In this paper, we strengthen the inapproximability result of \textsc{Min Co-secure Dom} for general graphs by showing that this problem can not be approximated within a factor of $(1- \epsilon)\ln |V|$ for perfect elimination bipartite graphs and star convex bipartite graphs unless \textsc{P=NP}. On the positive side, we show that  \textsc{Min Co-secure Dom} can be approximated within a factor of $O(\ln |V|)$ for any graph $G$ with $\delta(G)\geq 2$. 
For $3$-regular and $4$-regular graphs, we show that \textsc{Min Co-secure Dom} is approximable within a factor of $\dfrac{8}{3}$ and $\dfrac{10}{3}$, respectively. Furthermore, we  prove that \textsc{Min Co-secure Dom} is \textsc{APX}-complete for $3$-regular graphs. 
\end{abstract}

\section{Introduction}
Let $G=(V,E)$ be a finite, simple, and undirected graph with vertex set $V$ and edge set $E.$ The graph $G$ considered in this paper is without isolated vertices.  A set $D \subseteq V$ is said to be a \textit{dominating set} of $G$ if every vertex $v$ in $V\setminus D$ has an adjacent vertex $u$ in $D.$ The minimum cardinality among all dominating sets of $G$ is the \textit{domination number} of $G,$ and it is denoted by $\gamma(G).$ Given a graph $G$, in minimum dominating set problem (\textsc{Min Dom}), it is required to find a dominating set $D$ of minimum cardinality. \textsc{Min Dom} and its variations are studied extensively because of their real-life applications and theoretical applications. Detailed survey and results are available in \cite{haynes2017domination,haynes2013fundamentals,haynes2020topics}.

 A dominating set $S \subseteq V$ of $G=(V, E)$  is
called a secure dominating set of $G$, if $S$ is a dominating set of $G$ and for every $ u \in V \setminus S$ there exists a vertex
$v \in S,$ adjacent to $u$ such that $(S\setminus\{v\})\cup\{u\}$ is a dominating set of $G$. This important variation of domination was  introduced by Cockayne et al. \cite{cockayne2005protection}. The problem
of finding a minimum cardinality secure dominating set of a graph is known
as the Minimum Secure Domination Problem.  This problem and
its many variants have been extensively studied by several researchers \cite{araki2019secure,cockayne2005protection,
klostermeyer2008secure,kumar2020algorithmic,merouane2015secure,poureidi2021computing,wang2018complexity}. 

A set $S \subseteq V$ is a co-secure dominating set (\textsc{CSDS}) of a graph $G$ if $S$ is a dominating set and for each vertex $u \in S$ there exists a vertex $v \in V\setminus S$ such that $uv \in E$ and $(S\setminus \{u\}) \cup \{v\}$ is a dominating set of $G$. The minimum cardinality of a co-secure dominating set of $G$ is the co-secure domination number and it is denoted by $\gamma_{cs}(G).$ Given a graph $G=(V, E)$, in minimum co-secure dominating set  problem (\textsc{Min Co-secure Dom}), it is required to find a co-secure dominating set $S$ of minimum cardinality.
\textsc{Min Co-secure Dom} was introduced by Arumugam et al. \cite{arumugam2014co}, where they showed that the decision version of \textsc{Min Co-secure Dom} is \textsc{NP}-complete for  bipartite, chordal, and planar graphs.  They also determined the co-secure domination number for
some families of the standard graph classes such as paths, cycles, wheels, and complete $t$-partite graphs.   Some  bounds on the co-secure domination number for certain families of graphs were given by Joseph et al. \cite{joseph2018bounds}.  Manjusha et al. \cite{manjusha2020} characterized the Mycielski graphs with the co-secure domination number $2$ or $3$ and gave a sharp upper bound for $\gamma_{cs}(\mu(G))$,  where $\mu(G)$ is the Mycielski of a graph $G$. Later Zou et al.\cite{zou2022co}  proved that the co-secure domination number of proper interval graphs can be computed in linear time.  In \cite{arti2022co}, it is proved that \textsc{Min Co-secure Dom} is \textsc{NP}-hard to approximate within a factor of $(1 -\varepsilon)\ln |V|$  for any $\varepsilon >0$, and it is \textsc{APX}-complete for graphs with maximum degree 4.

In this paper, we extend the algorithmic study of \textsc{Min Co-secure Dom} by using certain properties of minimum double dominating set under some assumptions. The main contributions of the paper are summarised below. 

\begin{itemize}
\item   We prove that \textsc{Min Co-secure Dom}  can not be approximated  within a factor of $(1- \varepsilon)\ln |V|$ for perfect elimination bipartite graphs and star convex bipartite graphs unless \textsc{P}=\textsc{NP}. This improves the result due to Kusum and Pandey \cite{arti2022co}.

\item   We propose an approximation algorithm for  \textsc{Min Co-secure Dom}  for general graphs $G$ with $\delta(G)\geq 2,$  within a factor of $O(\ln |V|)$. In terms of maximum degree $\Delta$, it can be approximated within a factor of $2+2(\ln \Delta + 2)$.

\item For $3$-regular and $4$-regular graphs, we show that \textsc{Min Co-secure Dom} is  approximable within a factor of $\dfrac{8}{3}$ and $\dfrac{10}{3}$, respectively. 

\item We also prove that \textsc{Min Co-secure Dom} is \textsc{APX}-complete for $3$-regular graphs. 
\end{itemize}

\section{Preliminaries}
In this section, we give some pertinent definitions and state some preliminary results. Let $G=(V, E)$ be a finite, simple, and undirected graph with no isolated vertex. The open neighborhood of a vertex $v$ in $G$ is $N(v)=\{u \in V \mid uv \in E\}$ and the closed neighborhood is $N[v]=\{v\} \cup N(v).$ The degree of a vertex $v$ is $|N(v)|$ and is denoted by $d(v).$  If $d(v)=1$ then $v$ is called a \textit{pendant vertex} in $G$. 
The minimum degree and maximum degree of $G$ are denoted by $\delta$ and $\Delta,$ respectively. For $D\subseteq V,~G[D]$ denotes the subgraph induced by $D.$ We use the notation $[k]$ for $\{1, 2, \cdots, k\}.$ Given $S\subseteq V$ and $v\in S,$ a vertex $u\in V\setminus S$ is an $S$-external private neighbor ($S$-epn) of $v$ if $N(u)\cap S=\{v\}.$ The set of all $S$-epn of $v$ is denoted by $EPN(v, S).$ Some other notations and terminology which are not introduced here can be found in \cite{west2001introduction}.

A \textit{bipartite graph} is a graph $G=(V, E)$ whose vertices can be partitioned into two disjoint sets $X$ and $Y$ such that every edge has one endpoint in $X$ and other in $Y.$ We denote a bipartite graph with vertex bi-partition $X$ and $Y$ of $V$ as $G=(X, Y, E).$ The edge $uv\in E$ is a bi-simplicial edge if $N(u)\cup N(v)$ induces a complete bipartite subgraph in $G$. Let $\sigma = [e_1, e_2, \cdots, e_k]$ be an ordering of pairwise non-adjacent edges of $G.$ With respect to this ordering $\sigma$, we define $P_i,~ i \in [k]$ as the set of end vertices of the edges $\{e_1, e_2, \ldots, e_i\}$, and let $P_0= \emptyset$. The ordering $\sigma$ is said to be a \textit{perfect elimination ordering} for $G$ if $G[(X\cup Y)\setminus P_k]$ has no edge and each edge $e_i$ is bi-simplicial in $G[(X\cup Y)\setminus P_{i-1}].$ A graph $G=(V, E)$ is said to be a \textit{perfect elimination bipartite graph} if and only if it admits a perfect elimination ordering \cite{golumbic1978perfect}. A bipartite graph $G=(X, Y, E)$ is called a \textit{star convex bipartite graph} if a star graph $H=(X, E_X)$ can be  defined such that for every vertex $y\in Y,$ $N(y)$ induces a connected subgraph in $H$. 

\section{Approximation Algorithms}
In this section, we propose an approximation algorithm for \textsc{Min Co-secure Dom} whose approximation ratio is a logarithmic factor of the number of vertices of the input graph. To obtain the approximation ratio of \textsc{Min Co-secure Dom}, we require the approximation ratio of the minimum double dominating set problem  (\textsc{Min Double Dom}). Given a graph $G=(V, E)$, in \textsc{Min Double Dom}, the aim is to find a vertex set $D\subseteq V$ of minimum cardinality such that $|N(v) \cap D| \geq 2$, for all $v \in V \setminus D$. We shall denote $\gamma_2(G)$ as the cardinality of a minimum double dominating set in $G$. We will use the following proposition and a few lemmas to analyze our approximation algorithms' performance.

\begin{proposition} \label{CSDprop} (\cite{arumugam2014co})
Let $S$ be a \textsc{CSDS} of $G$. A vertex $v\in V\setminus S$ replaces $u\in S$ if and only if $v\in N(u)$ and $EPN(u, S)\subseteq N[v].$
\end{proposition}

\begin{lemma}\label{minimal_D_2}
If $G$ is a connected graph with at least 3 vertices then every minimal double dominating set $D_2$ of $G$ is a proper subset of $V$. Moreover, if $\delta(G) \geq 2$ then every minimal double dominating set $D_2$ is a co-secure dominating set of $G.$
\end{lemma}

\begin{proof}
Suppose there exists a minimal double dominating set $D_2$ of $G$ such that $|D_2|=|V|.$ Since $|V|\geq 3$ and $G$ is connected, there exists a vertex $v\in V$ with $d(v)\geq 2.$ Now, $D_2\setminus \{v\}$ is a double dominating set of $G$ contradicting the minimality of $D_2.$  

Let $D_2$ be a minimal double dominating set of $G$. From the minimality of $D_2$, it follows that every vertex $u \in D_2$ has at least one neighbor in $V\setminus D_2.$ Suppose there exists a vertex $p \in D_2$ such that $N(p) \subseteq D_2$. Then $D_2 \setminus \{p\}$ is also a double dominating set (as $d(p) \geq 2$.) This contradicts the minimality of $D_2$. 

Let $u$ be any vertex in $D_2$ and $v$ be its neighbor not in $D_2$. Next, we show that $S=(D_2\setminus \{u\})\cup \{v\}$ is a dominating set of $G.$ 
Suppose not, then there exists a vertex $w\in V\setminus S$ such that no vertex of $S$ dominates $w.$ $D_2$ is a dominating set of $G$ implies that $N(w)\cap D_2=\{u\}$. This contradicts the fact that $D_2$ is a double dominating set of $G.$ 

From the above arguments, it follows that $D_2$ is a co-secure dominating set of $G$.
\end{proof}

In the next lemma, we prove bounds on $\gamma_2(G)$ which we will use in designing approximation algorithms for \textsc{Min Co-secure Dom}.

\begin{lemma} \label{CSD_bound_with_D2}
For every graph $G$ with $\delta(G) \geq 2$, $\gamma_{cs}(G)\leq \gamma_2(G)\leq 2\gamma_{cs}(G).$ Moreover, these bounds are tight.
\end{lemma}
\begin{proof}
$\gamma_{cs}(G)\leq \gamma_2(G)$ holds as every minimal double dominating set of $G$ is also a \textsc{CSDS} of $G$ (by Lemma \ref{minimal_D_2}). Next we will prove that $\gamma_2(G)\leq 2\gamma_{cs}(G).$ Let $D$ be a $\gamma_{cs}$ set of $G.$ Let $D'=\{x\in D \mid EPN(x, D)\neq \emptyset \},$ and $D''=D\setminus D'.$ Let $A=\bigcup\limits_{x\in D'}^{}EPN(x,D).$ Then, every vertex $v\in (V\setminus \{D\cup A\})$ has at least two neighbors in $D''.$ By Proposition \ref{CSDprop}, for every vertex $x\in S$ there exists at least one vertex $x^*\in V\setminus S$ and $x^*\in EPN(x, S)$ such that $d_G(x^*) \geq |EPN(x, S)|$. Let $A'\subseteq A$ such that $A'$ contains exactly one vertex $x^*$ of each $EPN(x, D)$ for every $x\in D'.$ Thus, $|A'|=|D'|.$ Note that, every vertex in $A\setminus A'$ has at least two neighbors in $D'\cup A'.$ Let $B'$ be the smallest subset of $(V\setminus D)\setminus A'$ that dominates $D''.$ Since every vertex of $D''$ has $EPN(x, D'')=\emptyset,$ we obtain $|B'|\leq |D''|.$ Thus, $D\cup A'\cup B'$ is a double dominating set of $G.$ Hence, $\gamma_2(G)\leq |D|+|A'|+|B'|\leq |D|+|D'|+|D''|=2|D|=2\gamma_{cs}(G).$

These two inequalities are tight for the graphs $K_{2,2}$ and $K_n$ ($n \geq 3$), respectively.
\end{proof}

\begin{theorem} \label{D2-algo}
\textsc{Min Double Dom} can be approximated with an approximation ratio of $O(\ln |V|)$, where $V$ is the vertex set of the input graph $G$. It can also  be approximated within a factor of $1+ \ln (\triangle +2)$, where $\triangle$ is the maximum degree of $G$.
\end{theorem}
\begin{proof}
Given an instance $G=(V, E)$ of \textsc{Min Double Dom}, we construct a multiset multicover problem \cite{vazirani2001approximation} as follows. We take $V$ as the universe and for each vertex $v \in V$ we construct a multiset 
$S_v = N[v] \cup \{v\}$. In $S_v$, $v$ is appearing twice whereas other elements appear exactly once. We set the requirement of each vertex $v \in V$ as 2. Minimum Multiset Multicover problem can be approximated within a factor of $O(\ln |V|)$  (also $1+ \ln (\triangle +2)$) \cite{vazirani2001approximation}. Therefore, 
\textsc{Min Double Dom} can be approximated within a factor of $O(\ln |V|)$ (also $1+ \ln (\triangle +2)$).
\end{proof}

Next, we propose an algorithm (described in Algorithm \ref{Apx-CSD}) to compute an approximate solution of \textsc{Min Co-secure Dom}. This algorithm computes a minimal double dominating set $D_2$ of the input graph $G$ (with $\delta(G) \geq 2$) using the approximation algorithm described in Theorem \ref{D2-algo} and returns it as a \textsc{CSDS} of $G$. By Lemma \ref{minimal_D_2}, $D_2$ is also a \textsc{CSDS} of $G$. It is easy to observe that Algorithm \ref{Apx-CSD} runs in polynomial time.

\begin{algorithm2e}[H]
\textbf{Input:} A graph $G=(V,E).$\\
\textbf{Output:} A minimum \textsc{CSDS} of $G.$\\
\Begin{
Compute a double dominating set $D_2$ of $G$ (as described in Theorem \ref{D2-algo});\\
$S=D_2;$\\
 \Return $S$;}
 \caption{\textsc{Approx-CSD}}
 \label{Apx-CSD}
\end{algorithm2e}

\begin{theorem} \label{thm-ln-apx}
\textsc{Min Co-secure Dom} can be approximated within a factor of $O(\ln |V|)$, for graphs with $\delta(G) \geq 2$. It can also  be approximated within a factor of $2+ 2\ln (\triangle +2)$, where $\triangle$ is the maximum degree of $G$.
\end{theorem}
\begin{proof}
Let $S$ be the \textsc{CSDS} of $G$ computed by the Algorithm \ref{Apx-CSD}.  By Theorem \ref{D2-algo}, we have $|S| \leq O(\ln |V|)\gamma_2(G)$. Also, by Lemma \ref{CSD_bound_with_D2} we have $$|S| \leq O(\ln |V|) \gamma_2(G) \leq 2 O(\ln |V|) \gamma_{cs}(G) = O(\ln |V|) \gamma_{cs}(G).$$

Similarly, it can be observed that 
$|S|\leq [2+ 2\ln (\triangle +2)] \gamma_{cs}(G).$
\end{proof}

\section{Lower bound on approximation ratio}
In this section, we obtain a lower bound on the approximation ratio of \textsc{Min Co-secure Dom} for some subclasses of bipartite graphs. To obtain our lower bound, we establish an approximation preserving reduction from \textsc{Min Dom} to \textsc{Min Co-secure Dom}. We need the following lower bound result on \textsc{Min Dom}.

\begin{theorem}\label{chlebik-dinur}
(\cite{chlebik2008approximation,dinur2014analytical})
Unless \textsc{P=NP}, \textsc{Min Dom} can not be approximated within a factor of $(1-\varepsilon)\ln|V|$, for any $\varepsilon > 0$. Such a result holds for \textsc{Min Dom} even when restricted to bipartite graphs.
\end{theorem}

By using this theorem, we will prove similar lower bound results for \textsc{Min Co-secure Dom} for two subclasses of bipartite graphs, namely  perfect elimination bipartite graphs and star convex bipartite graphs.

\begin{theorem}\label{PEBG_CSD}
Unless \textsc{P=NP},
\textsc{Min Co-secure Dom} for a perfect elimination bipartite graph $G=(V,E)$ can not be approximated within $(1-\varepsilon)\ln |V|$, for any $\varepsilon>0$.
\end{theorem}
\begin{proof}
Given a graph $G=(V, E)$, an instance of \textsc{Min Dom}, we construct a graph $G'=(V', E')$, an instance of \textsc{Min Co-secure Dom}, as follows. Here we assume that $V=\{v_1, v_2, \ldots, v_n\}$. After making a copy of $G$, we introduce $n$ new vertices $a_1, a_2, \ldots, a_n$ and $n$ edges $v_ia_i$, for $i \in [n]$.Then we introduce 6 vertices $s, t, x, y, w, z$ and the edges $st, xy, wz$. Finally, we introduce the edge set $\{a_iv_i, v_is, a_ix, a_iz \mid i\in [n]\}$. It is easy to observe that $V'= V \cup \{a_i|i\in [n]\} \cup \{x,y,z,w, s,t\}$ and $E'=E\cup \{a_iv_i, v_is, a_ix, a_iz \mid i\in [n]\} \cup \{xy,zw, st\}$ and it is a polynomial time construction as $|V'|=2|V|+6$ and $|E'|=|E|+4|V|+3.$ $G'$ is a perfect elimination bipartite graph with the perfect elimination ordering $\{st,xy,zw, v_1a_1, v_2a_2,\cdots, v_na_n\}.$ For an illustration of this construction, we refer to Figure \ref{fig:pebg_csd}.

 \begin{figure}[htbp]
     \centering
     \includegraphics[width=9cm]{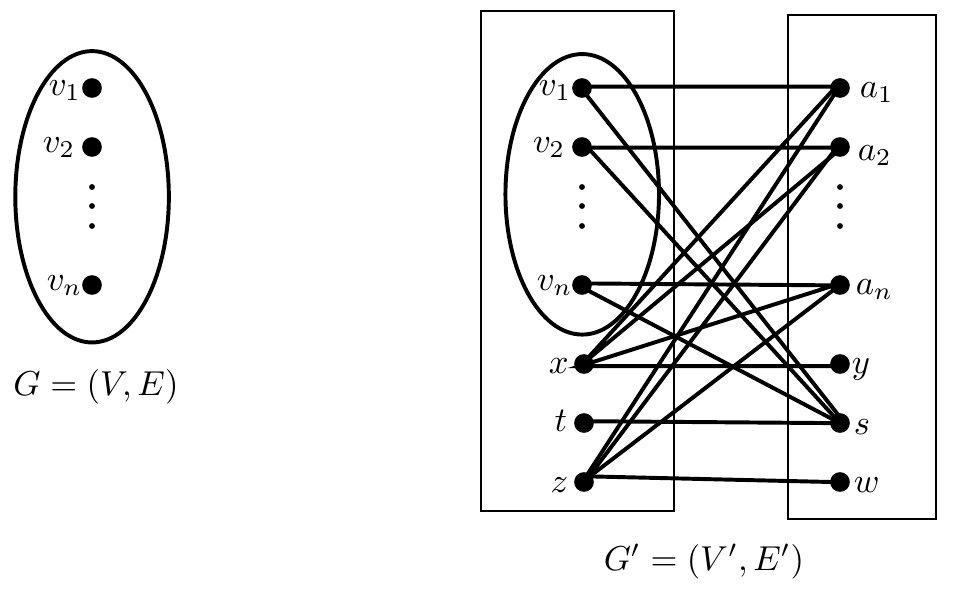}
     \caption{An illustration of the construction of $G'$ from $G$ in the proof of Theorem \ref{PEBG_CSD}}
     \label{fig:pebg_csd}
 \end{figure}

\begin{claim}\label{PEBG-lowerbound}
The graph $G$ has a dominating set of cardinality at most $k$ if and only if $G'$ has a \textsc{CSDS} of cardinality at most $k'=k+3$.
\end{claim}
\begin{proof}
Let $D$ be a minimal dominating set of $G$. It is easy to check that $S=D \cup \{x,z,s\}$ is a \textsc{CSDS} of $G'.$ Thus, $|S|=|D|+3.$

Conversely, let $S$ be a minimal \textsc{CSDS} of $G'$. $S \cap \{x, y\} = \{x\}$ as $y$ is the only degree 1 vertex adjacent to $x$. Similarly, $S\cap \{s, t\} = \{s\}$ and $S\cap \{w, z\} = \{z\}$. We will assume that $S$ does not contain any $a_i$ vertex. This is because, each $a_i$ vertex is dominated by at least two vertices $x$ and $z$, and if $a_i \in S$ then we will replace the vertex $a_i$ with $v_i$ in $S$. Now, we define $D= S \cap V$. If $D$ is a dominating set of $G$ then we are done. Otherwise, there exists a vertex $v_k$ which is not dominated by any vertex of $D.$ Now, $v_k$ is dominated only by $s\in S$ and $(S\setminus \{s\})\cup \{v\}$ is not a dominating set, for every $v \in (N_{G'}(s) \setminus S)$. This is a contradiction. Hence, $D$ is a dominating set of $G$ with $|S| = |D|+3$.
\end{proof}

Let us assume that there exists some (fixed) $\varepsilon >0$ such that \textsc{Min Co-secure Dom} for perfect elimination bipartite graphs with $|V'|$ vertices can be approximated within a ratio of $\alpha=(1-\varepsilon)\ln|V'|$ by a polynomial time algorithm $\mathbb A$. Let $l>0$ be a fixed integer with $l > \dfrac{1}{\varepsilon}$. By using algorithm $\mathbb A$, we construct a polynomial time algorithm for \textsc{Min Dom} as described in Algorithm \ref{Alg-Dom1}. 

Initially, if there is a minimum dominating set $D$ of $G$ with $|D|<l,$ then it can be computed in polynomial time. Since the algorithm $\mathbb A$ runs in polynomial time, the Algorithm \ref{Alg-Dom1} also runs in polynomial time. If the returned set $D$ satisfies $|D| < l$ then $D$ is a minimum dominating set of $G$ and we are done. 

Next, we will analyze the case when Algorithm \ref{Alg-Dom1} returned the set $D$ with $|D| \geq l$. By Claim  \ref{PEBG-lowerbound} we have $|S_o| = |D_o| + 3$, where $D_o$ and $S_o$ are minimum dominating set of $G$ and minimum \textsc{CSDS} of $G'$, respectively. Here $|D_o| \geq l$. 

\begin{algorithm2e}[H]
\textbf{Input:} A graph $G=(V,E).$\\
\textbf{Output:} A minimum dominating set $D$ of $G.$\\
\Begin{
        \uIf{there is a minimum dominating set $D$ of $G$ with $|D|<l$}{\Return $D$;}
        \Else{
        Construct the graph $G'$ as described above;\\
        Compute a \textsc{CSDS} $S$ in $G'$ using $\mathbb A$;\\
        $D=S\cap V$;\\
        \Return $D$;}}
\caption{\textsc{Approx-DOM1}}
\label{Alg-Dom1}
\end{algorithm2e}
Now,
$|D|\leq |S|-3 < |S| \leq \alpha |S_o| = \alpha (|D_o|+3)=\alpha (1+\frac{3}{|D_o|})|D_o|\leq \alpha (1+\frac{3}{l})|D_o|.$
This implies that Algorithm \ref{Alg-Dom1} approximates \textsc{Min Dom} within a ratio of $\alpha (1+\frac{3}{l}).$ Since $\frac{1}{l} < \varepsilon$
$$\alpha \bigg(1+\frac{3}{l}\bigg)\leq (1-\varepsilon)(1+3\varepsilon)\ln |V'|=(1-\varepsilon')\ln |V|,$$ where $\varepsilon'=3\varepsilon^2+2\varepsilon$ as $\ln |V'|=\ln(2|V|+6)\approx \ln |V|$ for sufficiently large value of $|V|.$

Therefore, Algorithm \ref{Alg-Dom1} approximates \textsc{Min Dom} within a ratio of $(1-\varepsilon)\ln|V|$ for some $\varepsilon>0.$ This contradicts the lower bound result in Theorem \ref{chlebik-dinur}.
 \end{proof}
 
Next, we prove the inapproximability of \textsc{Min Co-secure Dom} in star convex bipartite graphs by using the Theorem \ref{chlebik-dinur}.

\begin{theorem}\label{starconvex_CSD}
\textsc{Min Co-secure Dom} for a star convex bipartite graph $G=(V,E)$ can not be approximated within $(1-\varepsilon)\ln |V|$ for any $\varepsilon>0$, unless \textsc{P=NP.}
\end{theorem}
\begin{proof}
Given a bipartite graph $G=(X, Y, E)$, as an instance of \textsc{Min Dom}, we obtain a star convex bipartite graph $G'=(X', Y', E')$ such that $G$ has a dominating set of cardinality at most $k$ if and only if $G'$ has a \textsc{CSDS} of cardinality at most $k'=k+2.$ Now the construction of $G'$ from $G$ is as follows. After making a
copy of $G,$ we introduce four vertices $x_0, x, y_0, y$. Finally, we make every vertex of $X\cup \{x, x_0\}$ adjacent to $y$ and every vertex of $Y\cup\{y, y_0\}$ adjacent to $x.$ Now, $X'=\{X\}\cup \{x, x_0\}$, $Y'=\{Y\}\cup \{y, y_0\}$ and $E'=\{E\}\cup \{x_iy \mid x_i\in X\}\cup \{y_ix \mid y_i\in Y\}\cup \{x_0y, xy, xy_0\}.$ The new graph $G'=(V',E')$ formed from $G=(V,E)$ has $|V'|=|V|+4$ and $|E'|=|E|+n+3,$ which can be constructed in polynomial time. It can be observed that $G'$ is a star convex bipartite graph with the associated star graph which is shown in Figure \ref{starconvex_csd}.

 \begin{figure}[h]
     \centering
     \includegraphics[width=10cm, height=4cm]{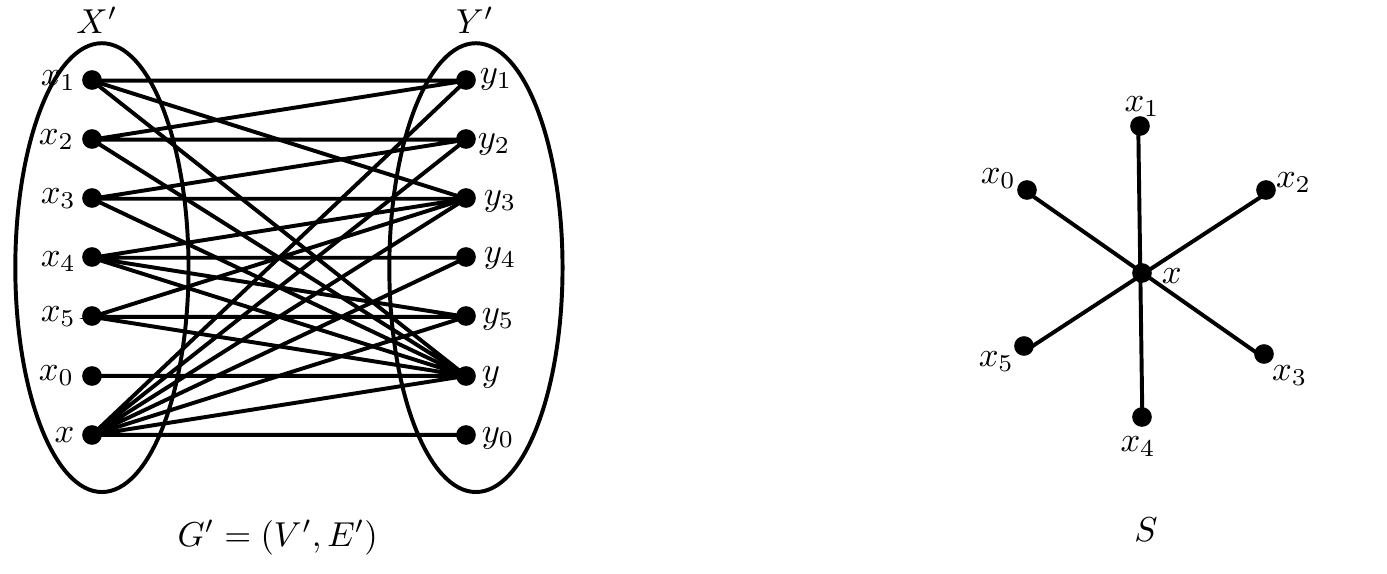}
     \caption{An illustration of the construction of $G'$ from $G$ in the proof of Theorem \ref{starconvex_CSD}}
     \label{starconvex_csd}
 \end{figure}

\begin{claim}\label{SCBG-Lowerbound}
$G$ has a dominating set of cardinality at most $k$ if and only if the graph $G'$ has a \textsc{CSDS} of cardinality at most $k'=k+2.$
\end{claim}
\begin{proof} 
 Suppose $D$ is a minimal dominating set of $G$ and let $S=D\cup \{x,y\}$. Clearly, $S$ is a \textsc{CSDS} of $G'$ with $|S|= |D|+2\leq k+2.$ Conversely, let $S$ be a minimal dominating set of $G'.$ Note that, $|S\cap \{x, y_0\}|=1$, and  similarly $|S\cap \{x_0, y\}|=1.$ If $x_0,y_0\in S,$ observe that $EPN(x_0,S)=y$ and $EPN(y_0,S)=x.$ So, without loss of generality, assume $\{x, y\}\subseteq S.$ Now, let $D=S\setminus \{x, y\}.$ Now we show that $D$ is a dominating set of $G.$ If $D$ is dominating set of $G,$ then we are done. Otherwise, suppose $D$ is not a dominating set of $G.$ Then there exists at least one vertex $v_k\in V(G)$ which is not dominated by any vertex of $D.$ Without loss of generality, assume $v_k\in X,$ then $v_k$ can only be dominated by $y\in Y'.$ Since $S$ is a \textsc{CSDS} of $G',$ $(S\setminus \{y\})\cup \{v_k\}$ is a  dominating set of $G',$ which is a contradiction. Thus, $D$ is a dominating set of $G$ of cardinality $|D|=|S|-2\leq k.$ Therefore, $G$ has a dominating set $D$ of cardinality at most $k$ if and only if $G'$ has a \textsc{CSDS} of cardinality at most $k'=k+2$. This completes the proof of this claim.
\end{proof}

Presume that there exists some (fixed) $\varepsilon >0$ such that \textsc{Min Co-secure Dom} for star convex bipartite graphs having $|V'|$ vertices can be approximated within a ratio of $\alpha=(1-\varepsilon)\ln|V'|$ by using an algorithm $\mathbb A$ that runs in polynomial time. Let $l>0$ be an integer. By using algorithm $\mathbb A$, we construct a polynomial time algorithm Algorithm \ref{Alg-Dom2} for \textsc{Min Dom}. 

\begin{algorithm2e}[H]
\textbf{Input:} A bipartite graph $G=(X, Y, E).$\\
\textbf{Output:} A minimum dominating set $D$ of $G.$\\
\Begin{
        \uIf{there is a minimum dominating set $D$ of $G$ with $|D|<l$}{\Return $D$;}
        \Else{
        Construct the graph $G'$ as described above;\\
        Compute a \textsc{CSDS} $S$ in $G'$ using the algorithm $\mathbb A$;\\
        $D=S\cap (X \cup Y)$;\\
        \Return $D$;}}
 \caption{\textsc{Approx-DOM2}}
 \label{Alg-Dom2}
\end{algorithm2e}

Firstly, if there is a minimum dominating set $D$ of $G$ with $|D|<l,$ then it can be computed in polynomial time. Moreover, Algorithm \ref{Alg-Dom2} runs in polynomial time as $\mathbb A$ runs in polynomial time.
Note that, if the returned set $D$ satisfies $|D| < l$ then it is a minimum dominating set of $G$ and we are done. Now, let us assume that the returned set $D$ satisfies $|D| \geq l$.

Let $D_o$ and $S_o$ be a minimum dominating set of $G$ and a minimum \textsc{CSDS} of $G',$ respectively. Then $|D_o|\geq l,$ and $|S_o|=|D_o|+2$ by the above Claim \ref{SCBG-Lowerbound}. Now,
$$ |D|\leq |S|-2 < |S| \leq \alpha |S_o| = \alpha (|D_o|+2)=\alpha \bigg(1+\frac{2}{|D_o|}\bigg)|D_o|\leq \alpha \bigg(1+\frac{2}{l}\bigg)|D_o|.$$
Hence, Algorithm \ref{Alg-Dom2} approximates \textsc{Min Dom} for given bipartite graph $G=(X, Y, E)$ within the ratio $\alpha (1+\frac{2}{l}).$ Let $l$ be the positive integer such that $\frac{1}{l}<\varepsilon.$ Then 
$$\alpha \bigg(1+\frac{2}{l}\bigg)\leq (1-\varepsilon)(1+2\varepsilon)\ln |X' \cup Y'|=(1-\varepsilon')\ln |X \cup Y|,$$ where $\varepsilon'=2\varepsilon^2-\varepsilon$ as $\ln |X' \cup Y'|=\ln(|X \cup Y|+4)\approx \ln |X \cup Y|$ for sufficiently large value of $|X \cup Y|.$

Therefore, Algorithm \ref{Alg-Dom2} approximates \textsc{Min Dom} within a ratio of $(1-\varepsilon)\ln|X \cup Y|$ for some $\varepsilon>0.$ This contradicts the lower bound result in Theorem \ref{chlebik-dinur}. \end{proof}

\section{Complexity on bounded degree graphs}
In this section, we show that \textsc{Min Co-secure Dom} is \textsc{APX-}complete for $3$-regular graphs. Note that the class \textsc{APX} is the set of all optimization problems which admit a $c$-approximation algorithm, where $c$ is a constant. From Theorem \ref{thm-ln-apx} it follows that \textsc{Min Co-secure Dom} can be approximated within a factor of $5.583$ for graphs with maximum degree at most $4$. We improve this approximation factor to $\dfrac{10}{3}$.

We first show that \textsc{Min Co-secure Dom} for $3$-regular graphs is approximable within a factor of $\dfrac{8}{3}$. 

\begin{algorithm2e}[H]
\textbf{Input:} A $3$-regular graph $G=(V,E)$. \\
\textbf{Output:} A \textsc{CSDS} $S$ of $G=(V,E).$\\
\Begin{
   $W'=\emptyset$;\\
   \While{$\exists$ an edge $uv\in E$}{
      $W'=W' \cup \{u,v\}$;\\
      Delete $N[u]\cup N[v]$ from $G$;
    }
    Let $T$ be the remaining vertices;\\
    $W=W'\cup T$;\\
    $S=V \backslash W$;\\
    \Return $S$;
}
\caption{\textsc{Approx-CSD-$3$RG}}
\label{algo-4}
\end{algorithm2e}

\begin{lemma} \label{lem-3-reg}
\textsc{Min Co-secure Dom} is approximable within a factor of $\dfrac{8}{3}$ for $3$-regular graphs.
\end{lemma}
\begin{proof}
 Let $S_o$ be a minimum \textsc{CSDS} of a $3$-regular graph $G=(V,E).$ A vertex $x\in S_o$ can co-securely dominate at most $3$ vertices of $V\setminus S_o$. Therefore, $|V\setminus S_o| \leq 3|S_o|$. This implies that 
 \begin{equation} \label{eq1}
 |S_o| \geq \frac{n}{4}.    
 \end{equation}
 The set $S$ of vertices returned by Algorithm \ref{algo-4} is a minimal double dominating set in $G$ because each vertex in $W'$ has exactly two neighbors in $S$. By Lemma \ref{minimal_D_2}, $S$ is a \textsc{CSDS} of $G$.

 Thus, $W=W'\cup T.$ Let $|W' \cup S|=n_1=n-|T|.$ Now $|W'|\geq \dfrac{n_1}{3},$ since in the while loop, the algorithm has picked two vertices and simultaneously removed at most six vertices from the graph.  Now, 
  \begin{equation*}
     |W|=|W'|+|T|\geq \frac{n_1}{3}+n-n_1\geq n-\frac{2n}{3}=\frac{n}{3}.
  \end{equation*}
  Thus, 
  \begin{equation}\label{eq3}
     |S|=|V|-|W|\leq n-\frac{n}{3}=\frac{2n}{3} 
  \end{equation}
  This yields the upper bound on the size of the \textsc{CSDS} returned. Combining equation (\ref{eq1}) and equation (\ref{eq3}), we obtain $\dfrac{|S|}{|S_0|}\leq \dfrac{8}{3},$ thereby proving the lemma.
\end{proof}

Next, we design a constant factor approximation algorithm for \textsc{Min Co-secure Dom} when the input graph is $4$-regular. 

\begin{algorithm2e}[H]
\textbf{Input:} A $4$-regular graph $G=(V, E)$. \\
\textbf{Output:} A \textsc{CSDS} $S$ of $G=(V, E).$\\
\Begin{
$W'=\emptyset$;\\
\While{$\exists$ a maximal induced path $P(u_1, u_k)=(u_1, u_2, \ldots, u_k)$ or an induced cycle $C=(u_1, u_2, \ldots, u_k, u_1)$}{
    $W'=W' \cup \{u_1, u_2, \ldots, u_k\}$;\\
    Delete the vertex set $\{u_1, u_2, \ldots, u_k\}$ and their neighbors from $G$;}
    Let $T$ be the remaining vertices;\\
    $W=W'\cup T$;\\
    $S=V \backslash W$;\\ 
    \Return $S$;}
 \caption{Approx-CSD-4RG}
 \label{Algo-5}
\end{algorithm2e}

\begin{lemma}
 \textsc{Min Co-secure Dom} for $4$-regular graphs can be approximated within a factor of $\dfrac{10}{3}$.
\end{lemma}
\begin{proof}
Given a $4$-regular graph $G$, in polynomial time Algorithm \ref{Algo-5} computes a vertex set $W$ such that the degree of each vertex in $G[W]$ is at most $2$.

\begin{claim}
    $S$ is a \textsc{CSDS} of $G$.
\end{claim}
\begin{proof}
    By Lemma \ref{minimal_D_2}, it is enough to show that $S$ is a minimal double dominating set of $G$.
    
    $S$ is a double dominating set of $G$ as each vertex in $W$ has at least two neighbors in $S$. Suppose $S$ is not a minimal double dominating set of $G$. Then there must be a vertex $v \in S$ such that $S \setminus \{v\}$ is a double dominating set of $G$. This implies that $v$ must have at least two neighbors in $S$. 
If $v \in S$ is adjacent to a vertex of degree two in $G[W]$ then $S\setminus \{v\}$ is not a double dominating set of $G$ (because $G$ is $4$-regular). This implies that $v$ must be adjacent to at least one end-vertex of an induced path $P$ in $G[W']$. This contradicts the maximality of $P$.
\end{proof}

Following the proof of Lemma \ref{lem-3-reg}, it can be proved that $|S_o| \geq \dfrac{n}{5}$. 
Let $W'$ be the set of vertices of degree $2$ in $G[W]$ and $Q = W \setminus W'$. By setting $n_1 = n -|Q|$ and following the proof of Lemma \ref{lem-3-reg}, it can be proved that $|W'| \geq \dfrac{n_1}{3}$. This implies that $|W| \geq \dfrac{n}{3}$ and $|S| \leq \dfrac{2n}{3}$. Therefore, $\dfrac{|S|}{|S_o|} \leq \dfrac{10}{3}.$ 
\end{proof}

 Before we prove that \textsc{Min Co-secure Dom} is \textsc{APX}-complete for $3$-regular graphs, we need some terminology and results regarding the partial monopoly set. 

\begin{definition}[\cite{peleg1996local}]\textsc{(Min Partial Monopoly Problem)}
Given a graph $G=(V,E),$ partial monopoly problem is to find a set $M\subseteq V$ of minimum cardinality such that for each $v\in V\setminus M,$ $|M\cap N[v]|\geq \dfrac{1}{2}|N[v]|.$
\end{definition}
It is known that for $3$-regular graphs \textsc{Min Partial Monopoly Problem} is \textsc{APX}-complete \cite{mishra2002hardness}.
It is easy to observe the following lemma:
\begin{lemma}\label{lemma_a}
    Let $G$ be a $3$-regular graph. A partial monopoly set $M$ of $G$ is a double dominating set of $G$ and vice versa. 
\end{lemma}

\begin{lemma}\label{lemma_b}
    Let $G$ be a $3$-regular graph and $S\subseteq V$ be a minimal \textsc{CSDS} of $G.$ In polynomial time one can construct a double dominating set $S'\subseteq V$ with $|S'|\leq 2|S|.$
\end{lemma}
\begin{proof}
    Let $S$ be a minimal \textsc{CSDS} of $G.$ Define $S_1$ be the set of vertices $v\in S$ such that $EPN(v,S)\neq \emptyset,$ and $S_2=S\setminus S_1.$ Now let $A=\bigcup\limits_{v\in S}^{} EPN(v,S).$ Note that every vertex in $A$ has exactly one neighbor in $S_1$ and every vertex in $(V\setminus S)\setminus A$ has at least two neighbors in $S_2$. By Proposition \ref{CSDprop}, for every vertex $x\in S$ there exists at least one vertex $x^*\in V\setminus S$ and $x^*\in EPN(x, S)$ such that $d_G(x^*) \geq |EPN(x, S)|$. Let us define a new set $A'\subseteq A,$ such that $A'$ contains that one vertex $x^*$ of each $EPN(x, S)$ for every $x\in S_1.$ Thus, $|A'|=|S_1|.$ Let $S'=S\cup A'.$ Now every vertex in $V\setminus S'$ has at least two neighbors in $S'.$ Hence $S'$ is a double dominating set of $G$ with cardinality $|S|+|A'|=|S|+|S_1|\leq 2|S|.$
\end{proof}

Now, we will prove that \textsc{Min Co-secure Dom} is \textsc{APX}-complete for $3$-regular graphs by establishing a reduction from \textsc{Min Partial Monopoly Problem} for $3$-regular graphs.

\begin{theorem}
    \textsc{Min Co-secure Dom} is \textsc{APX}-complete for $3$-regular graphs.
\end{theorem}
\begin{proof}
    Because of Lemma \ref{lem-3-reg}, it is enough to establish a polynomial time approximation ratio preserving reduction from \textsc{Min Partial Monopoly Problem} for $3$-regular graphs to \textsc{Min Co-secure Dom} for $3$-regular graphs.
    
    Given a $3$-regular graph $G=(V, E),$ an instance of \textsc{Min Partial Monopoly Problem}, we take the same graph $G$ as an instance of \textsc{Min Co-secure Dom}. Let $M_o$ be a minimum partial monopoly set of $G$ and $S_o$ be a minimum \textsc{CSDS} of $G.$ Then $|M_o|=\gamma_2(G)$ (by Lemma \ref{lemma_a}). Also, we have $|S_o|\leq |M_o|,$ by Lemma \ref{CSD_bound_with_D2}. Given a minimal \textsc{CSDS} $S$ of $G,$ we can construct a partial monopoly set $M\subseteq V$ with $|M|\leq 2|S|$ (from Lemma \ref{lemma_b} and \ref{lemma_a})   
    Therefore, $\dfrac{|M|}{|M_o|}\leq 2\dfrac{|S|}{|S_o|}.$ Hence, \textsc{Min Co-secure Dom} is \textsc{APX}-complete for $3$-regular graphs. 
\end{proof}

\section{Conclusion}
In this paper, we prove that \textsc{Min Co-secure Dom} is hard to approximate within a factor smaller than $\ln |V|$ for perfect elimination bipartite graphs and star convex bipartite graphs. On the positive side, we have proposed a $O(\ln |V|)$ approximation algorithm for \textsc{Min Co-secure Dom} for any graph. 
Apart from these, we have shown that for $3$-regular graphs and $4$-regular graphs \textsc{Min Co-secure Dom} admits a $\dfrac{8}{3}$ and $\dfrac{10}{3}$ factor approximation algorithms, respectively. It would be interesting to design a better approximation algorithm for $3$-regular graphs.  
We prove that it is \textsc{APX}-complete for $3$-regular graphs. We conjecture that it is \textsc{APX}-hard for $3$-regular bipartite graphs.

\end{document}